\documentclass[conference]{IEEEtran}
\usepackage{amssymb, amsmath, amsthm, color}
\usepackage{nicefrac}
\usepackage{subfigure}
\usepackage{graphicx}
\usepackage{pifont}
\usepackage{cite}
\usepackage{bbm}

\newcommand{\TL}{\delta}
\newcommand{\OL}{\omega}
\newcommand{\thpt}{\tau}
\newcommand{\smthness}{\sigma}
\newcommand{\xor}{\oplus}

\newcommand{\xmark}{\ding{55}}

\newtheorem{defn}{Definition}

\newtheorem{clm}{Claim}
\newtheorem{conj}{Conjecture}

\newtheorem{prop}{Proposition}

\title{Throughput-Smoothness Trade-offs in\\Multicasting of an Ordered Packet Stream}

\author{
\authorblockN{Gauri Joshi}
\authorblockA{EECS Dept.,
MIT\\
Cambridge, MA 02139, USA \\
Email: gauri@mit.edu}
\and
\authorblockN{Yuval Kochman}
\authorblockA{School of CSE, HUJI \\
Jerusalem, Israel \\
Email: yuvalko@cs.huji.ac.il}
\and
\authorblockN{Gregory W. Wornell}
\authorblockA{EECS Dept.,
MIT\\
Cambridge, MA 02139, USA \\
Email: gww@mit.edu}
}

\begin{document}
\maketitle
\begin{abstract}
An increasing number of streaming applications need packets to be strictly \emph{in-order} at the receiver. This paper provides a framework for analyzing in-order packet delivery in such applications. We consider the problem of multicasting an ordered stream of packets to two users over independent erasure channels with instantaneous feedback to the source. Depending upon the channel erasures, a packet which is in-order for one user, may be redundant for the other. Thus there is an inter-dependence between throughput and the smoothness of in-order packet delivery to the two users. We use a Markov chain model of packet decoding to analyze these throughput-smoothness trade-offs of the users, and propose coding schemes that can span different points on each trade-off. 

\end{abstract}

\section{Introduction}
\label{sec:intro}

There has been a rapid increase in streaming applications in both wired and wireless communication in recent years. Unlike traditional file transfer where only the total delay until the end of file transfer matters, streaming applications impose delay and order constraints on each individual packet in the file. Streaming includes audio/video applications VoIP, NetFlix, YouTube which play packets in-order. Although these applications need to play packets in-order, they can drop packets that experience large transmission delays. However other cloud-based applications such as remote desktop, Dropbox and Google Drive do not allow packet dropping, because packets represent instructions that need to be executed in-order. 

In \cite{gauri_isit_paper, gauri_infocom} we considered the problem of point-to-point streaming where the receiver applications require packets in-order. In several applications such as live video broadcasting, many users are accessing the content simultaneously. In this work we consider a multicast scenario where the source wants to ensure fast in-order packet delivery of a stream of packets to multiple users, while using the available bandwidth efficiently. 

The use of network coding in multicast packet transmission has been studied in \cite{fragouli_broadcast, delay_control_online_nw_coding, fu_sadeghi_medard, sundar_sadeghi_medard, shah_three_receiver}. The authors in \cite{fragouli_broadcast} use as a delay metric the number of coded packets that are successfully received, but do not allow immediate decoding of a source packet. For two users, the paper shows that a greedy coding scheme is throughput-optimal and guarantees immediate decoding in every slot. However, optimality of this scheme has not been proved for three or more users. In \cite{delay_control_online_nw_coding}, the authors analyze decoding delay with the greedy coding scheme in the two user case. However, both these delay metrics do not capture the aspect of in-order packet delivery. 

In-order packet delivery is considered in \cite{fu_sadeghi_medard, sundar_sadeghi_medard, shah_three_receiver}. These works consider that packets are generated by a Poisson process and are greedily added to all future coded combinations. 
In this work we provide a more general framework where the source can use feedback about past erasures to decide how many and which packets to add to the coded combinations, instead of just greedy coding over all generated packets. 

The main contribution of this work is to analyze how the priority given by the source to each user affects the in-order delivery of packets. We analyze the trade-off between the throughput and the smoothness of packet delivery for the two user case. In Section~\ref{sec:fixed_primary} we find the best coding scheme for a user that is piggybacking on a primary user that is always given higher priority. In Section~\ref{sec:greedy_scheme} we find the coding scheme that gives the best smoothness in packet delivery while ensuring throughput optimality to both users. In Section~\ref{sec:general_trade-off} we propose general coding schemes that can be used to tune the operating point on the throughput-smoothness trade-off of each user.

\section{Preliminaries} 
\label{sec:prelim}
\subsection{System Model}
\label{subsec:system_model}
Consider a source that has to multicast an infinite stream of packets $s_n$, $n \in \mathbb{N}$ of equal size to $K$ users $U_1, U_2, \cdots, U_K$. Time is divided into fixed length slots. In each slot the source transmits one coded linear combination of the source packets, with coefficients chosen from a large enough field to ensure independence of the coded combinations. 

We consider an i.i.d.\ erasure channel to each user such that every transmitted packet is received successfully at user $U_i$ with probability $p_i$, and otherwise received in error and discarded. The erasure events are independent across the users. The theoretical analysis presented in this paper focuses on the two user case. For simplicity of notation in this case, let $a \triangleq p_1 p_2$, $b \triangleq p_1 (1-p_2)$, $c \triangleq(1-p_1) p_2$ and $d \triangleq (1-p_1) (1-p_2)$, the probabilities of the four possible erasure patterns.\footnote{By considering other values of probabilities $a, b, c$ and $d=1-a-b-c$, our analysis can be extended to erasure events that are correlated across users but i.i.d\ across time slots.} We consider instantaneous and error-free feedback such that before transmission in slot $n$, the source knows about all erasures until slot $n-1$. 

The receiver-end application at each user requires packets strictly in order. Packets decoded out-of-order are buffered until the missing packets are decoded. Assume that the buffer is large enough to store all the out-of-order packets. Every time the earliest missing packet is decoded, a burst of in-order decoded packets is delivered to the application. For example, suppose that $s_1$ has been delivered and $s_3$, $s_4$, $s_6$ are decoded and waiting in the buffer. If $s_2$ is decoded in the next slot, then $s_2$, $s_3$ and $s_4$ are delivered to the application. 

%

\subsection{Performance Metrics}
\label{subsec:perf_metrics}

Ideally every user should get its next in-order packet, or its ``required" packet in every successful slot. The notion of required packets is formally defined as follows.
\begin{defn}[Required packet]
\label{defn:reqd_pkt}
The required packet of $U_i$ is its earliest undecoded packet after $n$ slots. Its index is denoted by $r_i(n)$, or $r_i$ where it is known without specifying $n$. 
\end{defn}
For example, if packets $s_1$, $s_3$ and $s_4$ have been decoded at user $U_i$, its required packet $s_{r_i}$ is $s_2$. 
 
Since the users experience independent channel erasures, the required packet of one user may be already decoded, and hence redundant for another user. Thus, there is a trade-off between the rate and the smoothness of packet delivery. We analyze this trade-off using the following performance metrics.

\begin{defn}[Throughput]
\label{defn:thpt}
The throughput $\thpt_i$ is the rate of in-order packet delivery to user $U_i$ and is given by
\begin{align}
\thpt_i = \lim_{n \rightarrow \infty} \frac{r_i(n)}{n} \quad \text{ in probability.}
\end{align}
\end{defn}

\begin{defn}[Smoothness Index]
\label{defn:smoothness}
The smoothness index $\smthness_i$ is the probability of a burst of in-order packets being delivered in a given slot. It is given by,
\begin{align}
\smthness_i = \lim_{n \rightarrow \infty} \frac{\sum_{k=1}^{n} \mathbbm{1}(r_i(k) > r_i(k-1)) }{n} \quad \text{ in probability,}
\end{align}
where $\mathbbm{1}(E)$ is the indicator function that is $1$ when event $E$ occurs and $0$ otherwise. 
\end{defn}

A higher $\thpt_i$ and $\smthness_i$ means a faster and smoother packet delivery respectively. The best possible trade-off is $(\thpt_i, \smthness_i) = (p_i, p_i)$. For the single user case, it can be achieved the simple Automatic-repeat-request (ARQ) scheme where the source retransmits the earliest undecoded packet until it is decoded. In this paper we aim to design coding strategies that maximize throughput and smoothness index for the two user case. 

In our analysis, it is convenient to express $\thpt_i$ and $\smthness_i$ in terms of the following intermediate quantities. 

\begin{defn}[Throughput Loss]
\label{defn:thpt_loss}
In $n$ slots, let $X_n$ be the number of unerased slots for $U_i$, and let $Y_n$ be the number of times the decoder at $U_i$ receives a combination of already decoded packets. Then its throughput loss $\TL_i$ is
\begin{align}
\TL_i = \lim_{n\rightarrow \infty} \frac{Y_n}{ X_n} \quad \text{ in probability}.
\end{align} 
\end{defn}
\begin{defn}[Order Loss]
\label{defn:order_loss}
In $n$ slots, let $Z_n$ be the number of times the decoder at $U_i$ receives an innovative combination, but cannot decode $s_{r_i}$. Then its order loss $\OL_i$ is
\begin{align}
\OL_i = \lim_{n\rightarrow \infty} \frac{Z_n}{ X_n} \quad \text{ in probability}.
\end{align} 
\end{defn}


The throughput $\thpt_i$ and smoothness index $\smthness_i$ of user $U_i$ can be expressed in terms of $\TL_i$ and $\OL_i$ as follows
\begin{align}
\thpt_i &= p_i (1-\TL_i) \label{eqn:lin_trans_1},\\
\smthness_i &= p_i (1-\TL_i-\OL_i) \label{eqn:lin_trans_2} .
\end{align}
In \eqref{eqn:lin_trans_1}, $p_i(1-\TL_i)$ is the fraction of slots in which the decoder receives an innovative coded combination. Since all packets are eventually delivered to the user, this is equal to $\thpt_i$. In \eqref{eqn:lin_trans_2}, $p_i(1-\TL_i-\OL_i)$ is the fraction of slots in which the required packet of $U_i$ is decoded. Since a burst of in-order packets is delivered to the user whenever this happens, this is equal to $\smthness_i$. Thus, the best achievable throughput-smoothness trade-off $(\thpt_i , \smthness_i) = (p_i, p_i)$ is equivalent to $(\TL_i, \OL_i) = (0,0)$. 
%
%
%
%
\subsection{Structure of Good Codes}

We now present code structures that maximize throughput and smoothness index of the users. 
%
\begin{clm}[Include only Required Packets]
\label{clm:reqd_pkts_only}
In a given slot, it is sufficient for the source to transmit a combination of packets $s_{r_i}$ for $i \in \mathcal{I}$ where $\mathcal{I}$ is some subset of $\{1,2, \cdots K\}$. 
\end{clm}

\begin{proof}[Proof]
Consider a candidate packet $s_c$ where $c \neq r_i$ for any $ 1 \leq i \leq K$. If $c < r_i$ for all $i$, then $s_c$ has been decoded by all users, and it need not be included in the combination. For all other values of $c$, there exists a required packet $s_{r_i}$ for some $i \in \{1,2, \cdots K\}$ that, if included instead of $s_c$, will allow more users to decode their required packets. Hence, including that packet instead of $s_c$ gives a lower order loss.

\end{proof}
\begin{clm}[Include only Decodable Packets]
\label{clm:inst_dec_only}
If a coded combination already includes packets $s_{r_i}$ with $ i \in \mathcal{I}$, and $U_j$, $j \notin I$ has not decoded all $s_{r_i}$ for $ i \in \mathcal{I}$, then a scheme that does not include $s_{r_j}$ in the combination gives a better throughput-smoothness trade-off than a scheme that does.
\end{clm}

\begin{proof}[Proof]
If $U_j$ has not decoded all $s_{r_i}$ for $ i \in \mathcal{I}$, the combination is innovative but does not help decoding an in-order packet, irrespective of whether $s_{r_j}$ is included in the combination. However, if we do not include packet $s_{r_j}$, $U_j$ may be able to decode one of the packets $s_{r_i}$, $i \in \mathcal{I}$, which can save it from an order loss in a future slot. Hence excluding $s_{r_j}$ gives a better throughput-smoothness trade-off.
\end{proof}

For the two user case, Claims~\ref{clm:reqd_pkts_only} and~\ref{clm:inst_dec_only} imply the following code structure.

\begin{prop}[Code Structure for the Two User Case]
\label{prop:two_user_code_struct}
Every achievable throughput-smoothness trade-off can be obtained by a coding scheme where the source transmits $s_{r_1}$, $s_{r_2}$ or the exclusive-or, $s_{r_1} \xor s_{r_2}$ in each slot. It transmits $s_{r_1} \xor s_{r_2}$ only if $r_1 \neq r_2$, and $U_1$ has decoded $s_{r_2}$ or $U_2$ has decoded $s_{r_1}$.
\end{prop}

In the rest of the paper we analyze the two user case and focus on coding schemes as given by Proposition~\ref{prop:two_user_code_struct}.
%

%
%
\begin{figure}[t]
\begin{center}
\begin{tabular}{ |c|c|c|c|}
  \hline
   Time & Sent &  $U_1$ & $U_2$\\
  \hline
     1 	&  $s_1$	& 	$s_1$ & \xmark \\ 
     2 	&  $s_2$	& 	\xmark & $s_2$ \\ 
     3 	&  $s_1 \xor s_2$	& 	$s_2$ & $s_1$ \\
     4 	&  $s_3$	& $s_3$ &\xmark\\
     5 	&  $s_4$	& $s_4$ &$s_4$ \\
   \hline
\end{tabular}
\caption{Illustration of the optimal coding scheme when the source always give priority to user $U_1$. The third and fourth columns show the packets decoded at the two users. Cross marks indicate erased slots for the corresponding user.\label{fig:fixed_primary_eg}}
\end{center}
\vspace{-0.65cm}
\end{figure}

\section{Optimal Performance for One of the Users}
\label{sec:fixed_primary}
%
%
In this section we consider that the source always gives priority to one user, called the primary user. We determine the best achievable throughput-smoothness trade-off for a secondary user that is ``piggybacking" on such a primary user.

\subsection{Coding Scheme}
Without loss of generality, suppose that $U_1$ is the primary user, and $U_2$ is the secondary user. Recall that ensuring optimal performance for $U_1$ implies achieving $(\thpt_1, \smthness_1) = (p_1, p_1)$, which is equivalent to $(\TL_1, \OL_1) = (0,0)$. While ensuring this, the best throughput-smoothness trade-off for user $U_2$ is achieved by the coding scheme given by Claim~\ref{clm:fixed_prim} below.  

\begin{clm}[Optimal Coding Scheme]
\label{clm:fixed_prim}
A coding scheme where the source transmits $s_{r_1} \xor s_{r_2}$ if $r_1 > r_2$ and $U_2$ has already decoded $s_{r_1}$, and otherwise transmits $s_{r_1}$, gives the best achievable $(\thpt_2, \smthness_2)$ trade-off while ensuring optimal $(\thpt_1, \smthness_1)$.
\end{clm}
\begin{proof}
Since $U_1$ is the primary user, the source must include its required packet $s_{r_1}$ in every coded combination. By Proposition~\ref{prop:two_user_code_struct}, if the source transmits $s_{r_1} \xor s_{r_2}$ if $U_2$ has already decoded $s_{r_1}$, and transmits $s_{r_1}$ otherwise, we get the best achievable throughput-smoothness trade-off for $U_2$.
\end{proof}
Fig.~\ref{fig:fixed_primary_eg} illustrates this scheme for one channel realization. 

\subsection{Markov Model of Packet Decoding}
\label{subsec:markov_model_fixed_prim}
Packet decoding at the two users with the scheme given by Claim~\ref{clm:fixed_prim} can be modeled by the Markov chain shown in Fig.~\ref{fig:fixed_primary_markov}. The state index $i$ can be expressed in terms of the number of gaps in decoding of the users, defined as follows. 

\begin{defn}[Number of Gaps in Decoding]
The number of gaps in $U_i$'s decoding is the number of undecoded packets of $U_i$ with indices less than $r_{\max} = \max_i r_i$. 
\end{defn}

In other words, the number of gaps is the amount by which a user $U_i$ lags behind the user that is leading the in-order packet decoding. The state index $i$, for $i \geq -1$ is equal to the number of gaps in decoding at $U_2$, minus that for $U_1$. Since the source gives priority to $U_1$, it always has zero gaps in decoding, except when there is a $c = p_2(1-p_1)$ probability erasure in state $0$, which causes the system goes to state $-1$. 
The states $i'$ for $i\geq 1$ are called ``advantage" states and are defined as follows.
\begin{defn}[Advantage State]
\label{defn:adv_state}
 The system is in an advantage state when $r_1 \neq r_2$, and $U_2$ has decoded $s_{r_1}$ but $U_1$ has not. 
\end{defn}
 \begin{figure}[t]
\centering
\includegraphics[width=2.5in]{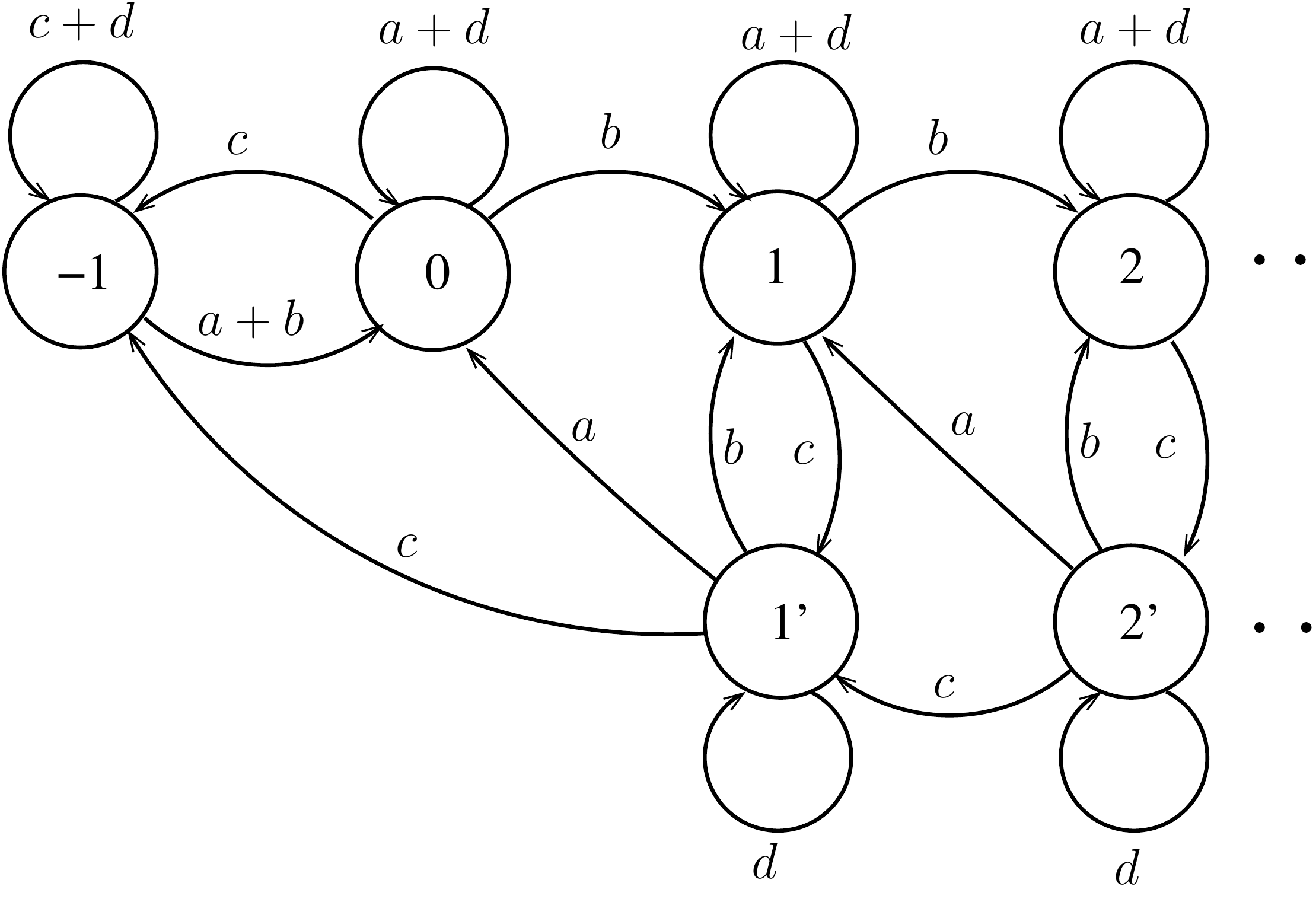}
\caption{Markov chain model of packet decoding with the coding scheme given by Claim~\ref{clm:fixed_prim}, where $U_1$ is the primary user. The state index $i$ represents the number of gaps in decoding of $U_2$ minus that for $U_1$. The states $i'$ are the advantage states where $U_2$ gets a chance to decode its required packet.\label{fig:fixed_primary_markov}}
\vspace{-0.6cm}
\end{figure}
By Claim~\ref{clm:fixed_prim}, the source transmits $s_{r_1} \xor s_{r_2}$ when the system is in an advantage state $i'$, and it transmits $s_{r_1}$ when the system is in state $i$ for $i \geq -1$. We now describe the state transitions of this Markov chain. First observe that with probability $d= (1-p_1)(1-p_2)$, both users experience erasures and the system transitions from any state to itself. When the system is in state $-1$, the source transmits $s_{r_1}$. Since $s_{r_1}$ has been already decoded by $U_2$, the probability $c=p_2(1-p_1)$ erasure also keeps the system in the same state. If the channel is successful for $U_1$, which occurs with probability $p_1 = a+b$, it fills its decoding gap and the system goes to state $0$.

The source transmits $s_{r_1}$ in any state $i$, $i \geq 1$. With probability $a = p_1 p_2$, both users decode $s_{r_1}$, and hence the state index $i$ remains the same. With probability $b = p_1(1-p_2)$, $U_1$ receives $s_{r_1}$ but $U_2$ does not, causing a transition to state $i+1$. With probability $c=(1-p_1) p_2$, $U_2$ receives $s_{r_1}$ and $U_1$ experiences an erasure due to which the system moves to the advantage state $i'$. When the system is an advantage state, having decoded $s_{r_1}$ gives $U_2$ an advantage because it can use $s_{r_1} \xor s_{r_2}$ transmitted in the next slot to decode $s_{r_2}$. From state $i'$, with probability $a$, $U_1$ decodes $s_{r_1}$ and $U_2$ decodes $s_{r_2}$, and the state transitions to $i-1$. With probability $c$, $U_2$ decodes $s_{r_2}$, but $U_1$ does not decode $s_{r_1}$. Thus, the system goes to state $(i-1)'$, except when $i=1$, where it goes to state $0$.

We now solve for the steady-state distribution of this Markov chain. Let $\pi_i$ and $\pi'_i$ be the steady-state probabilities of states $i$ for $i \geq -1$ and advantages states $i'$ for all $i \geq 0$ respectively. The steady-state transition equations are given by
\begin{align}
(1-a-d) \pi_i  &=b (\pi_{i-1} + \pi'_i) + a \pi'_{i+1} \quad \text{ for } i \geq 1 \label{eqn:steady_state_4},\\
(1-d) \pi'_i &= c( \pi_i + \pi'_{i+1}) \quad \quad \quad \quad \,\,\,\, \text{ for } i \geq 1 \label{eqn:steady_state_5}, \\
(1-c-d) \pi_{-1} &= c( \pi_0+ \pi'_{1}) \label{eqn:steady_state_1},\\
(1-a-d) \pi_0 &=a \pi'_1 + (a+b)\pi_{-1} \label{eqn:steady_state_2}.
\end{align}

By rearranging the terms in \eqref{eqn:steady_state_4}-\eqref{eqn:steady_state_2}, we get the following recurrence relation,
\begin{equation}
\pi_i =\frac{(1-a-d)}{c} \pi_{i-1} - \frac{b}{c}\pi_{i-2}  \quad \text{ for } i \geq 2.
\label{eqn:pi_recurrence}
\end{equation}
Solving the recurrence in \eqref{eqn:pi_recurrence} and simplifying \eqref{eqn:steady_state_4}-\eqref{eqn:steady_state_2} further, we can express $\pi_i$, $\pi'_i$ for $i \geq 2$ in terms of $\pi_1$ as follows,
\begin{align}
\frac{\pi_i}{\pi_{i-1}}&= \frac{b}{c},  \label{eqn:steady_state_recursion_1}\\
\frac{\pi'_i}{\pi_i} &=\frac{c}{a+c}. \label{eqn:steady_state_recursion_2}
\end{align}
From \eqref{eqn:steady_state_recursion_1} we see that the Markov chain will be positive-recurrent and a unique steady-state distribution exists only if $b<c$, which is equivalent to $p_1 < p_2$. The expressions of the steady-state probabilities are not given here due to space limitations. If $p_1 \geq p_2$, the expected recurrence time to state $0$, that is the time taken for $U_2$ to catch up with $U_1$ is infinity. 

\subsection{Throughput-smoothness Trade-off for the Secondary User}
\begin{figure}[t]
\centering
\includegraphics[width=3.5in]{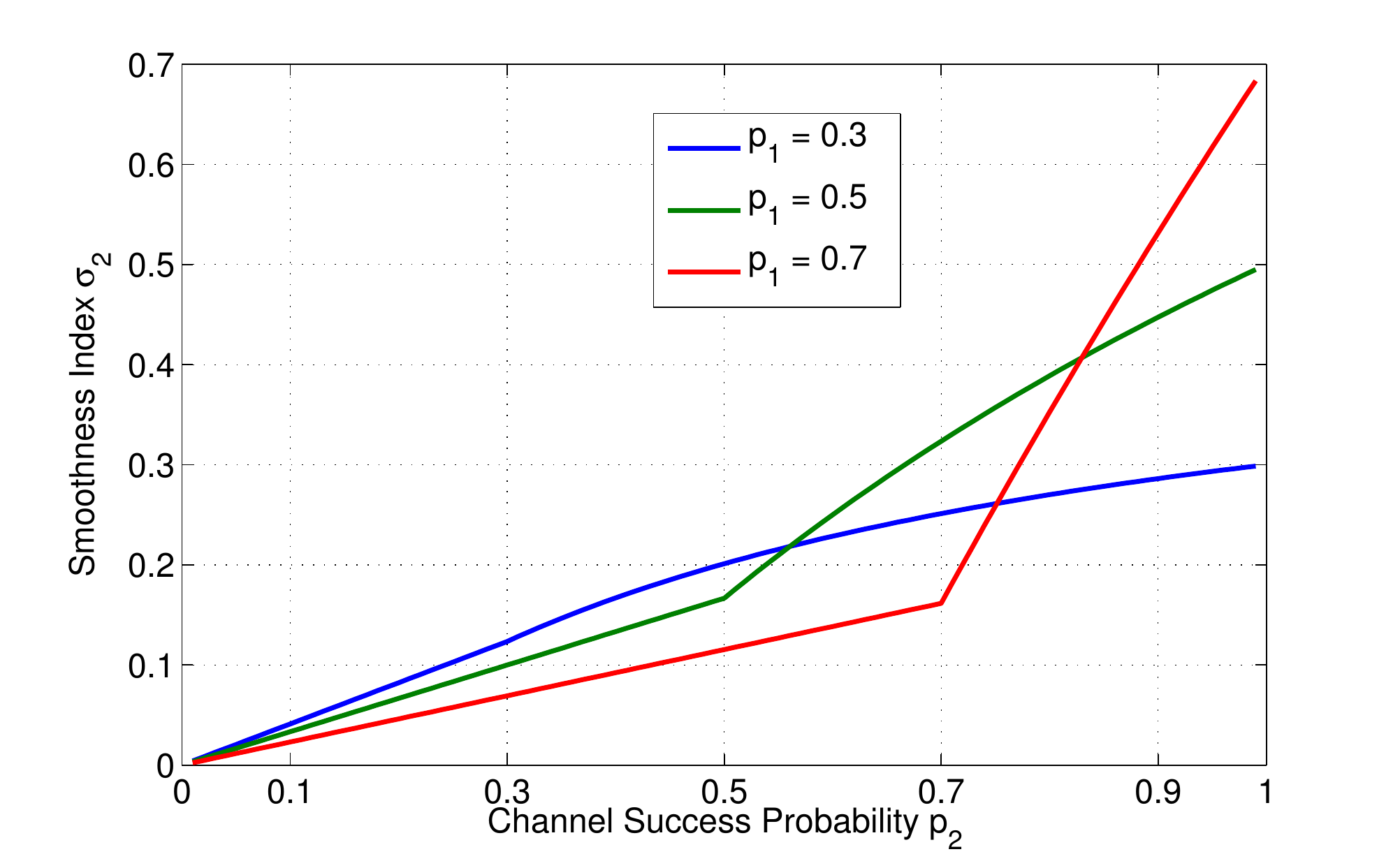}
\caption{Plot of the smoothness index $\smthness_2$ versus channel success probability $p_2$ of $U_2$ for different values of $p_1$. The increase in $\smthness_2$ with $p_2$ is sharper in the regime $p_2 >p_1$.\label{fig:thpt_in_order_fixed_markov}}
\vspace{-0.3cm}
\end{figure}
Since we always give priority to the primary user $U_1$, we have $(\TL_1, \OL_1) = (0,0)$. When $p_1 < p_2$, we can express the throughput loss $\TL_2$ and order loss $\OL_2$ in terms of the steady state probabilities of the Markov chain in Fig~\ref{fig:fixed_primary_markov}. User $U_2$ experiences an order loss when the system is in a state $i$, for $i \geq 1$ and the next slot is successful. And it experiences a throughput loss when it is in state $-1$ and the next slot is successful. By Definition~\ref{defn:thpt_loss} and Definition~\ref{defn:order_loss}, $\TL_2$ and $\OL_2$ are normalized by the success probability $p_2$. Thus, when $p_2 > p_1 $,
\begin{align}
(\TL_2, \OL_2) &= \left(\pi_{-1}, \sum_{i=1}^{\infty} \pi_i \right), \\
&= \left( \frac{c-b}{a+c},  \frac{b(a+b)}{c(1-d)} \right), \\
&= \left(1 - \frac{p_1}{p_2} , \frac{p_1^2(1-p_2)}{p_2(1-p_1)(p_1+ p_2-p_1 p_2)} \right). \label{eqn:T2_O2_stable}
\end{align}
If $p_1 > p_2$, the system drifts infinitely to the right side and hence it is in state $i$ or $i'$ for $ i \geq 1$ with probability $1$. Using \eqref{eqn:steady_state_recursion_2}, we can show that the intermediate quantities are
\begin{align}
(\TL_2, \OL_2) &= \left(0,  \frac{a+c}{a+2c}\right) = \left( 0 , \frac{1}{2-p_1} \right). \label{eqn:T2_O2_unstable}
\end{align}
Both \eqref{eqn:T2_O2_stable} and \eqref{eqn:T2_O2_unstable} converge to the same values when $p_1 = p_2$. 

Using \eqref{eqn:lin_trans_1} and \eqref{eqn:lin_trans_2} we can express $\thpt_2$ and $\smthness_2$ in terms of the intermediate quantities as follows,
\begin{align}
\thpt_2 &= \min(p_1, p_2) \label{eqn:thpt_2_fixed_prim} \\
\smthness_2 &=\begin{cases} 
 \frac{p_2(1-p_1)}{2-p_1} &\text{ for } p_2 < p_1, \\
 \frac{p_1( p_2(1-p_1)-p_1^2(1-p_2))}{(1-p_1)(p_1+ p_2-p_1 p_2)}& \text{ for } p_2 \geq p_1 .
\end{cases}
\end{align} 
Fig.~\ref{fig:thpt_in_order_fixed_markov} shows the smoothness index $\smthness_2$ as a function of channel success probability $p_2$ of the secondary user, for different values of $p_1$. We observe that the increase of $\sigma_2$ with $p_2$ is much faster when $p_2>p_1$. However, when $p_2 >p_1$, the throughput saturates at $p_1$ as given by \eqref{eqn:thpt_2_fixed_prim}. This implies that secondary user piggybacking on a fixed primary user achieves good smoothness without much throughput loss when $p_2$ is slightly more than $p_1$.

For $K >2$ users with a fixed priority order, we can find lower bounds on the throughput and smoothness index of a user. This can be done by fusing users with higher priority into a super user, thus reducing it to the two user case.

\section{Throughput Optimality for Both Users}
\label{sec:greedy_scheme}
We now consider the case where the source wants to ensure throughput optimality to both users, and we determine the best achievable smoothness indices of the users.
\begin{figure}[t]
\begin{center}
\begin{tabular}{ |c|c|c|c|}
  \hline
   Time & Sent &  $U_1$ & $U_2$\\
  \hline
     1 	&  $s_1$	& 	$s_1$ & \xmark \\ 
     2 	&  $s_2$	& 	\xmark & $s_2$ \\
     3 	&  $s_1 \xor s_2$	& \xmark &$s_1$\\
     4 	&  $s_3$	& $s_3$ & \xmark \\
     5 	&  $s_2 \xor s_3$	& $s_2$ &$s_3$ \\
   \hline
\end{tabular}
\caption{Illustration of the greedy coding scheme in Definition~\ref{defn:thpt_greedy_coding}. The third and fourth columns show the packets decoded by the two users. Cross marks indicate erased slots for the corresponding user.\label{fig:greedy_eg}}
\end{center}
\vspace{-0.5cm}
\end{figure}

\subsection{Greedy Coding Scheme}

Let $r_{\text{max}} = \max(r_1, r_2)$ and $r_{\text{min}} = \min (r_1 ,r_2)$, where $r_1$ and $r_2$ are the indices of the required packets of the two users. We refer to the user(s) with the higher index $r_i$ as the leader(s) and the other user as the lagger. Thus, $U_1$ is the leader and $U_2$ is the lagger when $r_1 > r_2$, and both are leaders with $r_1 = r_2$. 

\begin{defn}[Greedy Coding]
\label{defn:thpt_greedy_coding}
In greedy coding, the source transmits $s_{r_{\text{max}}} \xor s_{r_{\text{min}}}$ when the lagger has decoded $s_{r_{\text{max}}}$ but the leader has not, and transmits $s_{r_{\text{max}}}$ otherwise.
\end{defn}

Fig.~\ref{fig:greedy_eg} illustrates the greedy coding scheme and the sequence of packets decoded by the two users.

\begin{clm}[Optimality of Greedy Coding]
\label{clm:opt_greedy}
The greedy coding scheme in Definition~\ref{defn:thpt_greedy_coding} gives the best smoothness indices $\smthness_1$ and $\smthness_2$ while ensuring throughput optimality to both users.
\end{clm}
\begin{proof}
To ensure throughput optimality, the source must include packet $s_{r_{\text{max}}}$ in the coded combination. By Proposition~\ref{prop:two_user_code_struct}, the coding scheme that includes $s_{r_{\text{min}}}$ in the combination only if the lagger has decoded $s_{r_{\text{max}}}$ but the leader has not, maximizes the smoothness indices of the two users. 
\end{proof} 
%

\subsection{Markov Analysis of Packet Decoding}
Packet decoding with greedy coding can be modeled by the Markov chain shown in Fig.~\ref{fig:greedy_markov}, which is a two-sided version of the Markov chain in Fig.~\ref{fig:fixed_primary_markov}. User $U_1$ is the leader in states $i \geq 1$ and $U_2$ is the leader in states $ i \leq -1$, and both are leaders in state $0$. The system is in the advantage state $i'$ if packet is decoded by the lagger but not the leader. All the transitions for states $i$ and $i'$ for $i \geq 1$ are same as in Fig~\ref{fig:greedy_markov}. The transitions for states $i \leq -1$ are symmetric to $i \geq 1$, with probabilities $b$ and $c$ interchanged. 

For the states $i$, with $i \geq 2$ the steady state recursions are same as \eqref{eqn:steady_state_recursion_1} and \eqref{eqn:steady_state_recursion_2}. Similarly for states $i \leq -2$ we have
\begin{align}
\frac{\pi_{-i}}{\pi_{-i+1}} &=\frac{c}{b},\label{eqn:steady_state_recursion_3}\\
\frac{\pi'_{-i}}{\pi_{-i}} &= \frac{b}{a+b}.\label{eqn:steady_state_recursion_4}
\end{align}

The right hand side of the chain is transient if $b > c$ (equivalent to $p_1 >p_2$), and the left side is transient if $b<c$. Hence, the Markov chain is transient when $b \neq c$, which is equivalent to $p_1 \neq p_2$, and null-recurrent when $p_1 = p_2$. 

\begin{figure}[t]
\centering
\includegraphics[width=3.5in]{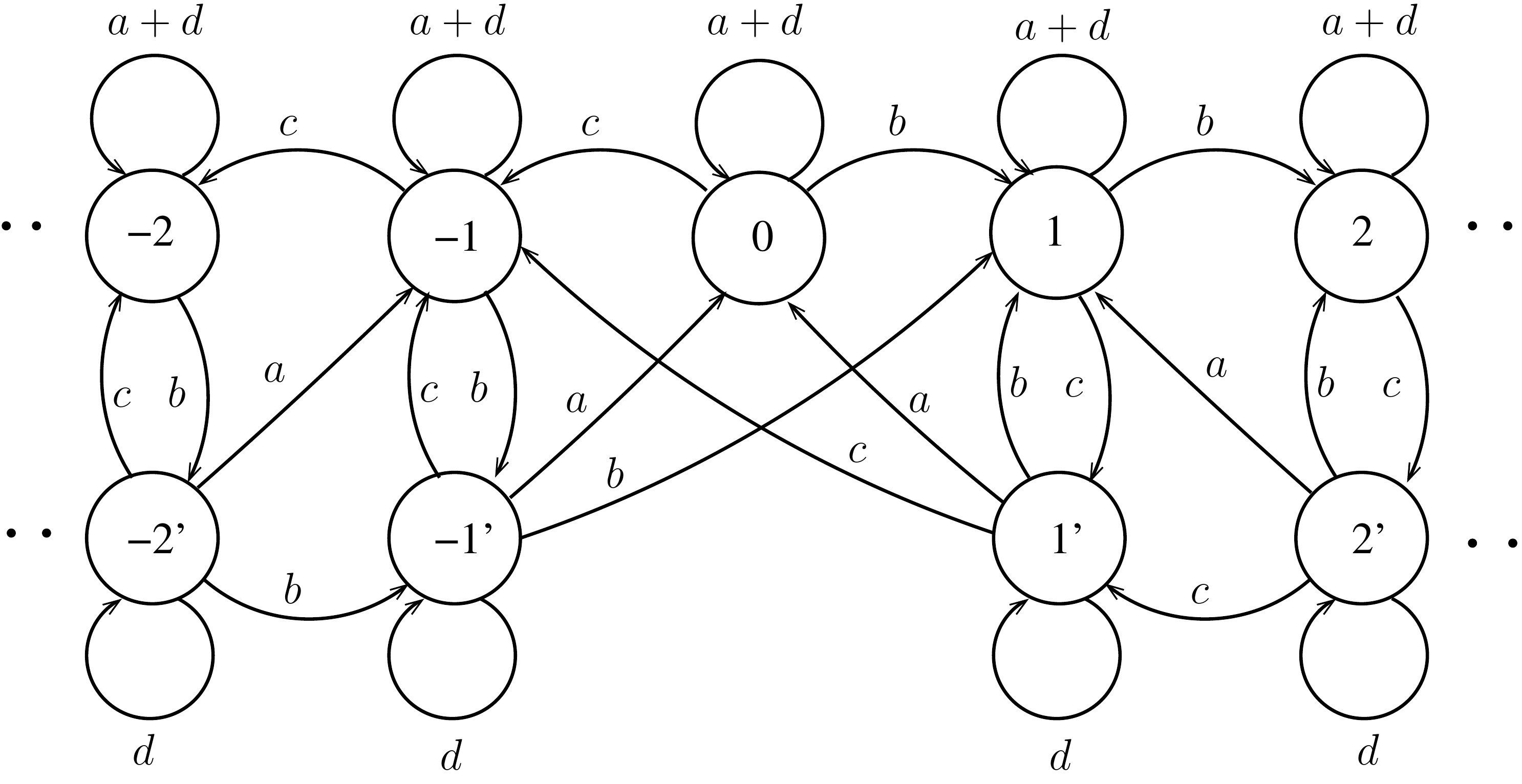}
\caption{Markov chain model of packet decoding with the greedy coding scheme in Definition~\ref{defn:thpt_greedy_coding} that ensures throughput optimality to both users. It is a two-sided version of the Markov chain in Fig.~\ref{fig:fixed_primary_markov}.\label{fig:greedy_markov}}
\vspace{-0.35cm}
\end{figure}

\subsection{Throughput-Smoothness Trade-offs of the two users}
By Claim \ref{clm:opt_greedy} the greedy coding scheme ensures throughput optimality for both users. Hence, $\thpt_1 = p_1$, $\thpt_2 = p_2$ and the throughput losses $\TL_1$ and $\TL_2$ are both zero.

When the right side of chain is transient, that is, if $p_1 > p_2$, 
\begin{align}
\OL_1 &= 0,\label{eqn:greedy_p_1_higher_1}\\
\OL_2 &=\sum_{i=1}^{\infty} \pi_i= \frac{1}{2-p_1}.\label{eqn:greedy_p_1_higher_2}
\end{align}
If $p_2 > p_1$, $\OL_2 = 0$ and $\OL_1$ is same as in \eqref{eqn:greedy_p_1_higher_2} with $p_1$ replaced by $p_2$. When $p_1 = p_2$ the Markov chain is null-recurrent, in which case, the limit in Definition~\ref{defn:order_loss} does not converge, and hence $\OL_1$ and $\OL_2$ are ill-defined.   

Using \eqref{eqn:lin_trans_2}, we can express the smoothness trade-off $\smthness_2$ as,
\begin{align}
\smthness_2 &= \begin{cases}
p_2 &\quad \quad \text{ if } p_1 < p_2, \\
\frac{p_2}{2-p_1}&\quad \quad \text{ if } p_1 > p_2,\\
\text{undefined} & \quad \quad \text{otherwise}.
\end{cases}
\label{eqn:TB_greedy}
\end{align}
The expression $\smthness_1$ is same as \eqref{eqn:TB_greedy} with the probabilities $p_1$ and $p_2$ interchanged everywhere. We can see that the user with the better channel gets the optimal smoothness index, but at the cost of the other user experiencing a much lower smoothness in packet delivery.

\section{General Throughput-Smoothness Trade-offs}
\label{sec:general_trade-off}
For the general case, we propose coding schemes that can be combined to tune the priority give to each user and achieve different points on its throughput-smoothness trade-off. 

\subsection{Proposed Codes}
We can modify the greedy coding scheme in Definition~\ref{defn:thpt_greedy_coding} to get more general schemes called $(N,M)$ multicast codes. 
\begin{defn}[$(N,M)$ Multicast Codes]
\label{defn:N_M_codes}
In the $(N,M)$ multicast code, the source follows the greedy coding scheme in Definition~\ref{defn:thpt_greedy_coding}, except in states $i \geq N$ or $i \leq -M$ for $N, M \geq 1$, of the model in Fig.~\ref{fig:greedy_markov}, where it gives priority to the lagger and transmits $s_{r_{\text{min}}}$.
\end{defn}
\begin{figure}[t]
\centering
\includegraphics[width=3.5in]{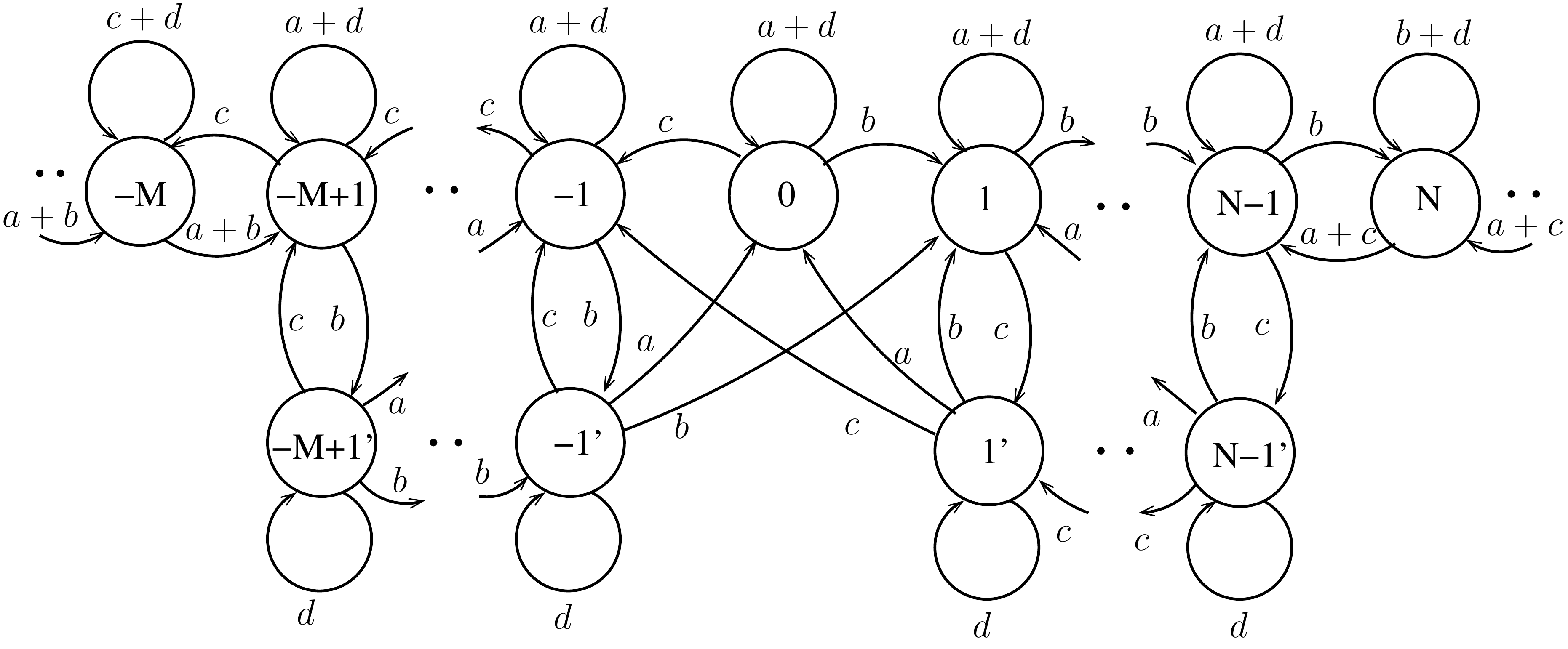}
\caption{Markov chain model of packet decoding with the $(N,M)$ multicast code. In this scheme we give priority to the lagger and transmit its required packet when the system is in state $i \geq N$ or $i \leq  -M$. \label{fig:fixed_prim_M_N_markov}}
\vspace{-0.35cm}
\end{figure}
The packet decoding using the $(N,M)$ multicast code can be modeled by the Markov chain in Fig.~\ref{fig:fixed_prim_M_N_markov}. It is a generalization of the chain in Fig.~\ref{fig:greedy_markov}. By Definition~\ref{defn:N_M_codes}, the source gives priority to the lagger $U_2$ and transmits $s_{r_2}$ when the system is in state $i \geq N$, and it transmits $s_{r_1}$ when in states $i \leq -M$, where $U_1$ is the lagger. Thus, there are backward state transitions with probabilities $p_2 = a+c$ from states $i \geq N$ and forward transitions from states $i \leq -M$ with probability $p_1 = a+b$. All other state transitions are same as the greedy coding scheme in Fig.~\ref{fig:greedy_markov}. There is no change in the coding strategy for the advantage states because by Claim~\ref{clm:inst_dec_only} it is always optimal to transmit $s_{r_1} \xor s_{r_2}$ in these states. 


By varying $N$ and $M$, we can tune the level of priority to each user. For example, we can give higher priority to $U_1$ by decreasing $M$ (or increasing $N$), keeping the other parameter fixed. The coding scheme in Claim~\ref{clm:fixed_prim} is the $(\infty, 1)$ code, the greedy coding scheme in Definition~\ref{defn:thpt_greedy_coding} is the $(\infty, \infty)$ code. 

We can combine different $(N,M)$ codes to get more general coding schemes. 
Suppose the $(N_r, M_r)$ codes, for $r= 1, 2, \cdots, R$ achieve the throughput-smoothness trade-offs $(\thpt_i^{(r)}, \smthness_i^{(r)})$, for $r= 1, 2, \cdots, R$ respectively. Then we can achieve any convex combination of these trade-offs as follows. 

\begin{clm}[Time-sharing $(N,M)$ codes]
\label{clm:combining_M_N_codes}
For every $W$ slots, if the source uses the $(N_r,M_r)$ code for $x_r W$ slots, with $0\leq x_r \leq 1$ and $\sum_{r=1}^{R} x_r = 1$, then for large enough $W$,
\begin{align}
(\thpt_i, \smthness_i) = \left( \sum_{r=1}^{R} x_r \thpt_i^{(r)}, \sum_{r=1}^{R} x_r \smthness_i^{(r)} \right).
\end{align}
\end{clm}
The proof is omitted due to space limitations.
\begin{conj}
\label{conj:rand_policy}
A randomized policy where in every slot, the source uses the $(N_r,M_r)$ code with probability $x_r$, for $r = 1,2,\cdots R$ gives the same $(\thpt_i, \smthness_i)$ trade-off each user $U_i$ as the time-sharing policy in Claim~\ref{clm:combining_M_N_codes}.
\end{conj}
Conjecture~\ref{conj:rand_policy} implies that the trade-off of a scheme where in state $i$, the source transmits $s_{r_{\text{min}}}$ with probability $q_i$, can be expressed as a combination of the trade-offs of $(N,M)$ codes. 
%
%
%
\subsection{Throughput-Smoothness Trade-offs of the Two Users}
We can solve for the following recursive relations between the steady-state probabilities of the chain in Fig.~\ref{fig:fixed_prim_M_N_markov}. For $N \geq 2$,
\begin{align}
\frac{\pi_{i}}{\pi_{i-1}} &= \begin{cases} \frac{b}{c}&\quad\quad \text{for } 2 \leq i \leq N-2, \\
 \frac{b(1-d)}{c(a+c)} &\quad\quad \text{for } i = N-1,\\
 \frac{b}{(a+c)} &\quad\quad \text{for } i = N, \\
 0 & \quad \quad \text{for } i > N.
\end{cases}
\label{eqn:rho_recursion}
\end{align}
\begin{align}
\frac{\pi'_{i}}{\pi_i} = \begin{cases} 
\frac{c}{a+c}  &\quad\quad \text{ for } 1 \leq i \leq N-2, \\
\frac{c}{1-d} &\quad\quad \text{ for } i = N-1, \\
0 &\quad\quad \text{ for } i \geq N.
\end{cases}
\label{eqn:mu_recursion}
\end{align}
For states $i \leq -1$, the ratios $\pi_{-i}/\pi_{-i+1}$ and $\pi'_{-i}/\pi_{-i}$ are same as given by \eqref{eqn:rho_recursion} and \eqref{eqn:mu_recursion}, but with the probabilities $b$ and $c$ interchanged, and $N$ replaced by $M$. The throughput and order losses can be expressed in terms of $\pi_i$ as
\begin{align}
(\TL_1, \OL_1) &=  \left(\pi_{N} , \sum_{i=1}^{M-1}\pi_{-i} \right), \label{eqn:M_N_T1_O1} \\
(\TL_2, \OL_2) &= \left(\pi_{-M}, \sum_{i=1}^{N-1} \pi_i \right).\label{eqn:M_N_T2_O2}
\end{align}
\begin{figure}
\centering
\includegraphics[width=3.5in]{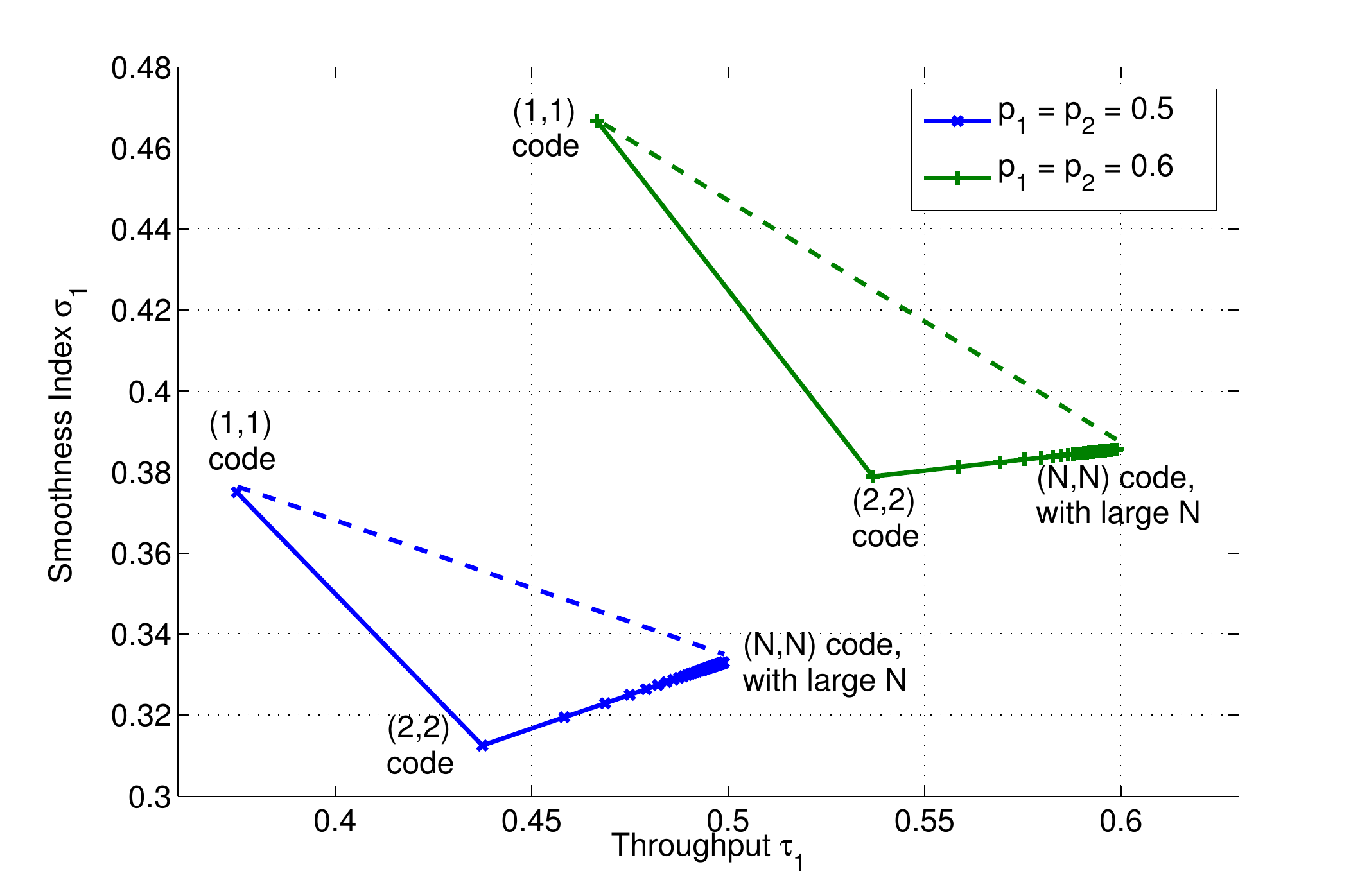}
\caption{Plot the smoothness index $\smthness_1$ versus throughput $\thpt_1$ for $(N,N)$ codes, with $p_1 = p_2$. Due to symmetry $\smthness_2 = \smthness_1$ and $\thpt_2 = \thpt_1$. Time-sharing between the $(1,1)$ code and $(N,N)$ code for large $N$ gives the best trade-off.\label{fig:M_N_result_plot}}
\end{figure}

Using \eqref{eqn:lin_trans_1} and \eqref{eqn:lin_trans_2} we can express the throughput-smoothness trade-off $(\thpt_i, \smthness_i)$ in terms of $(\TL_i, \OL_i)$. From \eqref{eqn:M_N_T1_O1} and \eqref{eqn:M_N_T2_O2} we see that to achieve high smoothness index (proportional to $-(\TL_i+\OL_i)$ ) for one user, one needs to sacrifice on the smoothness index of the other user. 
%
%
In Fig.~\ref{fig:M_N_result_plot} we plot the throughput-smoothness tradeoff $(\thpt_1, \smthness_1)$ for $(N,N)$ codes with $N \geq 1$, and when $p_1 = p_2$. We choose $N=M$ to give equal priority to both users. Since the two sides of the Markov chain are symmetric for these parameters, $\thpt_1 = \thpt_2$ and $\smthness_1 = \smthness_2$. 
Using Claim~\ref{clm:combining_M_N_codes} we can infer that time-sharing between the $(1,1)$ and $(N, N)$ codes for large $N$ gives the best throughput-smoothness trade-off, which is the dashed line joining the end points of each curve in Fig.~\ref{fig:M_N_result_plot}. Further, if Conjecture~\ref{conj:rand_policy} is true, then randomized policy which uses the $(1,1)$ code with probability $q$ for $0 \leq q \leq 1$, and uses the $(N,N)$ code with large $N$ otherwise can also achieve this trade-off. From this we can infer that in a coding strategy which favors the non-leader with probability $q_i$ in state $i$, setting all $q_i$ to the same value $q$ gives the best trade-off.  


%
%
\section{Concluding Remarks}
\label{sec:conclu}

In this work we study the trade-off between the throughput and smoothness in packet delivery when the application requires the packets \emph{in-order}. We use a Markov chain model to analyze the trade-off when the source is multicasting a packet stream to two users over erasure channels with instantaneous feedback. The throughput and smoothness index achieved by a user depends on the priority given to it by the source. By considering the cases of fixed and greedy priority we can show that both users cannot simultaneously achieve optimal throughput and optimal smoothness. We propose general coding schemes that can be used to tune the priority given to each of the users and thus span different points of their throughput-smoothness trade-offs. Future directions include extending this framework to more users, and analyzing second-order delay characteristics such as the exponent of inter-delivery delay.

\bibliographystyle{ieeetr}
\bibliography{streaming}

\end{document}